\documentclass[a4paper,UKenglish,cleveref, autoref, thm-restate]{lipics-v2021}
\hideLIPIcs

\usepackage{microtype}
\graphicspath{{./Figures/}}

\usepackage[labelformat=simple]{subcaption}
\theoremstyle{definition}

\newcommand{\eps}[0]{\varepsilon}
\newcommand{\seq}[1]{\ensuremath{(#1_i)_{i=0}^n}}

\title{Simply Realising an Imprecise Polyline is NP-hard}

\author{Thijs van der Horst}{Department of Information and Computing Sciences, Utrecht University, the Netherlands\and Department of Mathematics and Computer Science, TU Eindhoven, the Netherlands}{t.w.j.vanderhorst@uu.nl}{}{}
\author{Tim Ophelders}{Department of Information and Computing Sciences, Utrecht University, the Netherlands\and Department of Mathematics and Computer Science, TU Eindhoven, the Netherlands}{t.a.e.ophelders@uu.nl}{}{}
\author{Bart van der Steenhoven}{Department of Mathematics and Computer Science, TU Eindhoven, the Netherlands}{b.j.v.d.steenhoven@student.tue.nl}{}{}

\authorrunning{T. van der Horst, T. Ophelders, and B. van der Steenhoven}


\nolinenumbers

\begin{document}
\maketitle
\begin{abstract}
    We consider the problem of deciding, given a sequence of regions, if there is a choice of points, one for each region, such that the induced polyline is simple or weakly simple, meaning that it can touch but not cross itself.
    Specifically, we consider the case where each region is a translate of the same shape.
    We show that the problem is NP-hard when the shape is a unit-disk or unit-square.
    We argue that the problem is NP-complete when the shape is a vertical unit-segment.
\end{abstract}

    \section{Introduction}
        Finding a planar drawing of a graph has been studied widely due to its many applications, for example in floor-planning and visualisations~\cite{planargraphbook}.
        Often, when the to be drawn graphs arise from geometric data, such as road networks, each vertex already has a predetermined location.
        In this case, if we insist on drawing edges as straight line segments, there is a unique corresponding drawing.
        Real-world data such as GPS trajectories are often imprecise.
        Imprecision is often modeled by assigning to each vertex an imprecision region, and an assignment of each vertex to a point in its imprecision region is called a \emph{realisation}.
        
        Algorithms for imprecise data have been considered in the past~\cite{dataimprecisionthesis, structimprecise}. 
        A particularly relevant setting for us is when we know that the underlying data has no self-intersections.
        We can then look for a realisation whose drawing also has no self-intersections.
        Godau~\cite{embeddinginaccuracies} showed that this problem is NP-hard when the imprecision regions are closed disks of varying radius.
        In follow-up work, Angelini et al.~\cite{anchoredgraphs} showed that the problem remains NP-hard for straight-line drawings when the regions are all equal-sized disks, diamonds or squares under the requirement that these regions do not overlap. 
        In this work, we consider this problem in a very restricted graph class, namely path-graphs.
        That is, we are looking for a realisation that is a simple polyline.
        Algorithms exist to determine whether a polygon or polyline is weakly simple~\cite{recogwspol, detectwspol}.
        The most relevant related work is that of Löffler~\cite{imprecisetours}, who proved the NP-hardness of determining whether a realisation exists that is a weakly simple polygon in the case that the imprecision regions are scaled translates of a fixed shape, such as a segment or circle.
        Finally, Silveira et al.~\cite{noncrossingpaths} showed that the problem is NP-complete for straight-line drawings of graphs where each vertex has degree exactly one (i.e. the graph is a matching) using unit-length vertical segments as imprecision regions.
        
        \subparagraph*{Definitions and problem statement.}
        A \emph{polygonal chain} or \emph{polyline} is a piecewise-linear path in the plane, and is given by a sequence $P=\seq{p}$ of points in the plane, called vertices.
        The pieces of $P$ are the $n$ segments connecting consecutive vertices.
        With slight abuse of notation, we interpret $P$ as the polyline, instead of its vertices.
        An \emph{imprecise polyline} is given by a sequence of regions $R=\seq{r}$ in the plane.
        We call the polyline $\seq{p}$ a \emph{realisation} of $R$ if $p_i\in r_i$ for each $i$.
        
        A polyline is \emph{simple} (or plane) if it does not cross or touch itself.
        An \emph{$\eps$-perturbation} of a polyline $P=(p_i)_{i=0}^n$ is a polyline $P'=(p'_i)_{i=0}^n$ such that $\|p_i-p'_i\|\leq\varepsilon$ for all $i$.
        A polyline $(p_i)_{i=0}^n$ is \emph{weakly simple} if for every $\eps > 0$ there is an $\eps$-perturbation that is simple.
        We are interested in the problem of deciding whether an imprecise polyline $R$ admits a weakly simple realisation.
        

        \subparagraph*{Results and organisation.} 
        We show that the problem is NP-hard already when each region is a translate of the same shape $S$.
        In \Cref{app:circles} we prove this for the case where $S$ is a unit disk, and discuss how the construction can be adapted to show NP-hardness for unit squares.
        It follows that it is hard to compute the infimum $\eps$ for which a polyline has a simple $\eps$-perturbation under the $L_1$, $L_2$, and $L_\infty$ norms.
        In \Cref{app:vertsegs} we show NP-completeness when $S$ is a unit-length vertical segment.
    
    \section{Unit disks as regions} \label{app:circles}
        We prove by reduction from planar monotone 3SAT that the problem is NP-hard when all regions are unit disks.
        We first introduce planar monotone 3SAT.
        A monotone 3CNF formula is a formula~$\varphi(x_1,\dots,x_m):=\bigwedge_i C_i$ with Boolean variables $x_1,\dots,x_m$, where each clause $C_i$ is either positive (of the form $(x_j\lor x_k\lor x_l)$) or negative (of the form $(\neg x_j\lor\neg x_k\lor\neg x_l)$).
        A monotone rectilinear layout of a monotone 3CNF formula is a plane drawing of the incidence graph between its clauses and variables, in which variables are drawn as disjoint horizontal segments on the $x$-axis, all positive clauses lie above the $x$-axis, all negative clauses lie below the $x$-axis, and edges are rectilinear, do not cross the $x$-axis, and have at most one bend. We may assume for each clause that its central edge does not bend, see \Cref{afig:plan3sat-vb}.
        Planar monotone 3SAT is an NP-complete problem that asks, given a monotone 3CNF formula $\varphi$ with a given monotone rectilinear layout, whether $\varphi$ is satisfiable.
        For our reduction we construct gadgets consisting of imprecise polylines corresponding to variables, clauses, and edges of the layout.
        We then show how to connect these gadgets into a single imprecise polyline that has a weakly simple realisation if and only if the formula $\varphi$ if satisfiable.
        
        \begin{figure}[ht]
            \centering
            \includegraphics{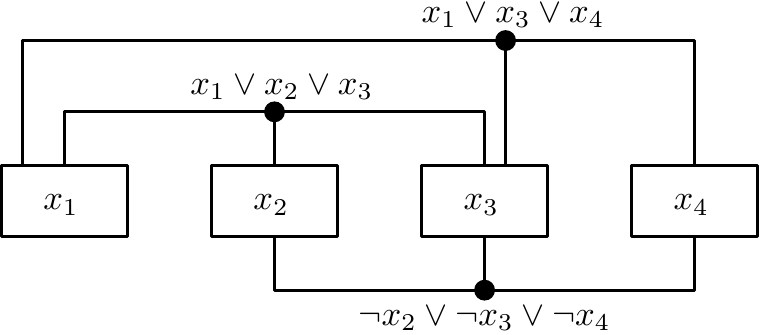}
            \caption{Example of an induced variable-clause graph for planar monotone 3SAT.}
            \label{afig:plan3sat-vb}
        \end{figure}
        
        \subparagraph*{Pivot gadget.}
        \begin{figure}[b]
            \centering%
            \subcaptionbox{\label{afig:x-gadget-og}}[.5\textwidth]%
                {\includegraphics[page=1]{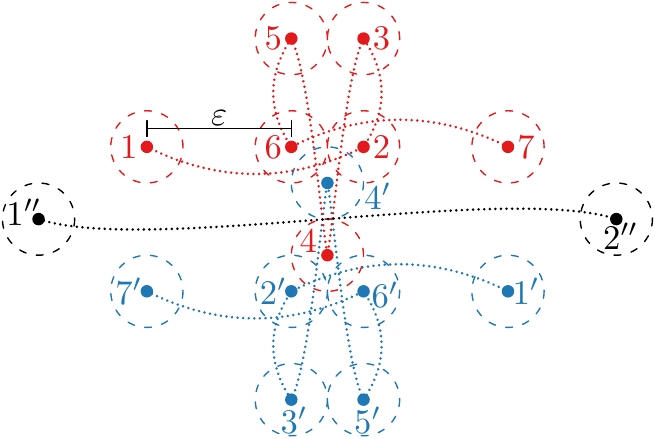}}%
            \subcaptionbox{\label{afig:x-gadget-straight}}[.5\textwidth]%
                {\includegraphics[page=2]{X-gadget/X-gadget}}%
            \caption{\subref{afig:x-gadget-og} The pivot gadget. \subref{afig:x-gadget-straight} A weakly simple realisation, showing a simple perturbation. Note that this perturbation moves most vertices outside of their assigned disks. }
            \label{afig:x-gadget}
        \end{figure}
        Before we construct the gadgets for the variables, clauses and edges, we introduce an auxiliary gadget: the \emph{pivot gadget}.
        The purpose of the pivot gadget is to force a given edge of any weakly simple realisation to go through some fixed point of our choosing.
        The gadget is illustrated in \Cref{afig:x-gadget} and consists of three components.
        Two of these components together form the pivot so that the edge of the third component must pass through the center of the pivot.
        If we wish to force this edge to pass through $(0,0)$, the coordinates of the top component will be
        \[1, \ldots, 7 \mapsto (-1-\eps,2) , (1,2) , (1,5) , (0,-1) , (-1,5) , (-1,2) , (1+\eps,2)\] 
        for some value $\eps>0$, and the bottom component is a rotated copy of the top component:
        \[1', \ldots, 7' \mapsto (1+\eps,-2) , (-1,-2) , (-1,-5) , (0,1) , (1,-5) , (1,-2) , (-1-\eps,-2)\text{.}\] 
        There is some freedom in choosing the placement of regions $1''$ and $2''$, but to make sure that the realisation passes through the pivot we require that the x-coordinate of region $1''$ is less than -1, that of region $2''$ is greater than 1, and the tangent lines to the pair of circles $1''$ and $2''$ all intersect the segment between $(0,-5)$ and $(0,5)$.
        We can accommodate any angle of the edge between regions $1''$ and $2''$ by choosing $\eps>0$ sufficiently small. Namely, such that the point $(-\eps, 2)$ where we will realise region 1 lies above the tangent lines from point $(0,0)$ to region $1''$, and such that symmetric properties hold for regions 7, $1'$ and $7'$. 

        \Cref{afig:x-gadget-straight} depicts a weakly simple realisation of the pivot gadget.
        Other weakly simple realisations may differ in placement for points in regions 1, 7, $1'$, $7'$, $1''$ and $2''$ only, but the segment connecting $1''$ and $2''$ must always pass through $(0,0)$.
        
        \begin{lemma} \label{alem:gate}
            For any weakly simple realisation, the segment of the pivot gadget connecting regions $1''$ and $2''$ must pass through $(0, 0)$. 
        \end{lemma}
        \begin{proof}
            First we consider the top pivot component individually. If we look at the part of any polyline consisting of only points 3, 4 and 5 then we note that these will always form a V-shape where the angle can be changed by changing the location of the points. From the point of view of points 1, 2, 6 and 7 the lines from 3 to 4 and from 4 to 5 might as well be infinite since with the defined region placements a line from 1 to 2 can never pass above or below these lines. This means the lines from 3 to 4 and from 4 to 5 define (closed) half-planes. In order to avoid intersections points 1 and 2 must be in the same half-plane for every valid placement. Since any point chosen for region 1 will always be in the left half-plane, we know that the point of region 2 must also be to the left. At the same time, similar reasoning tells us that the points of region 6 and 7 must both be in the right half-plane for both lines between 3, 4 and 5. With the given region placements, the only way to make this happen is if the points in region 2, 3, 4, 5 and 6 are all placed on a vertical line, which must be $x=0$ since this is the only line that hits all of these regions. The same holds for the mirrored bottom pivot component. This means that points 4 and $4'$ lie on the same vertical line. We must have that point 4 has y-coordinate greater or equal than points $4'$ because otherwise the line between the $1''$ and $2''$ will always have an intersection. The only realisation where point 4 is not below point $4'$ is one where they are both placed at $(0,0)$. Since the tangent lines to the pair of circles $1''$ and $2''$ all intersect the segment between $(0,-5)$ and $(0,5)$,  any line between these two regions cannot pass completely above or below the pivot gadget. Thus, for any weakly simple realisation, the point $(0,0)$ then becomes the only position where an edge can pass from one side of the pivot gadget to the other without intersections.
        \end{proof}
        
        \begin{lemma}
            Any choice of points for regions $1''$ and $2''$ such that the segment connecting them passes through $(0,0)$ has a corresponding weakly simple realisation of the pivot gadget. 
        \end{lemma}
        \begin{proof}
            Observe that, similar as before, the points of regions 2 to 6 and $2'$ to $6'$ can be placed on a vertical line such that an edge from region $1''$ to $2''$ through $(0,0)$ does not intersect any of them. Furthermore, we chose $\eps$ such that the regions 1, 7, $1'$ and $7'$ do not cause any intersections. To see this consider region 1, for which we can choose point $(-\eps, 2)$ in a realisation. We have chosen $\eps$ such that this point lies above the tangent lines from point $(0,0)$ to region $1''$, which means it lies above any edge between region $1''$ and $2''$ that goes through $(0,0)$. Therefore the edge from region 1 to 2 does not intersect the edge from $1''$ to $2''$, and the same holds for the edges from the other three corner points of the gadget. 
        \end{proof}
        
        \subparagraph*{Variable gadget.}
        For the reduction we want to have a variable gadget that can take on two discrete states, one for true and one for false, such that the difference between them can be used to allow for different behaviour. To this end, we devised our variable gadget as seen in \Cref{afig:var-gadget-og}. The exact coordinates and order of the points is as follows, assuming the leftmost point as our $(0, 0)$ and $l$ as an arbitrary value such that $l > 8$:
        \[1, \ldots, 6 \mapsto (0,0) , (8, 0) , (5, 2) , (5, -2) , (2, 0) , (l, 0) \]
        
        \begin{lemma} \label{alem:vargadget}
        There are exactly two distinct weakly simple realisations for the regions of the variable gadget.
        \end{lemma}
        \begin{proof} To see why this is the case first note that for any realisation of the regions, the only way a line from region 5 to 6 does not cross the line from region 3 to 4 is if it goes above or below it. The same holds for the line from region 1 to 2. Due to the placement of the regions, these lines can only go above the line from 3 to 4 if the point of 3 is placed at the very bottom $(5,1)$ of its region and both endpoints of the line going over it are placed at the very top of their regions at $y=1$. A similar property holds for when a line needs to go below the line between 3 and 4. Finally, note that putting the points of regions 1, 2, 5 and 6 all on the same horizontal line does not lead to a weakly simple polyline. 
        \end{proof}
        If the final horizontal line from region 5 to 6 is above the blocking vertical line between regions 3 and 4, then the gadget is in its false state as seen in \Cref{afig:var-gadget-false}. We consider the situation where it is below as the true state (\Cref{afig:var-gadget-true}). 
        
        \begin{figure}[ht]
            \centering
            \subcaptionbox{\label{afig:var-gadget-og}}[.33\textwidth]%
                {\includegraphics{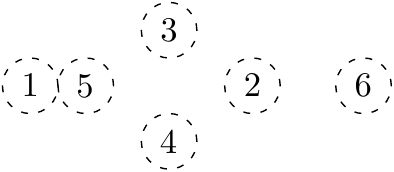}}%
            \subcaptionbox{\label{afig:var-gadget-false}}[.33\textwidth]%
                {\includegraphics{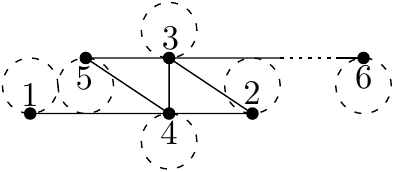}}%
            \subcaptionbox{\label{afig:var-gadget-true}}[.33\textwidth]%
                {\includegraphics{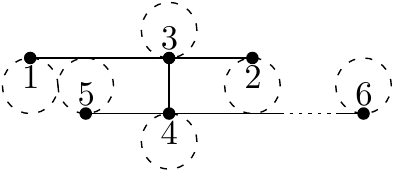}}%
            \caption{\subref{afig:var-gadget-og} The variable gadget. \subref{afig:var-gadget-false} The false state. \subref{afig:var-gadget-true} The true state.}
            \label{afig:variables}
        \end{figure}
        
        \subparagraph*{Clause gadget.} 
        The clause gadget is used to represent the disjunctions of the shape $(x_1 \vee x_2 \vee x_3)$, where the variables can also all be negated. The initial placement of the regions for this gadget can be seen in \Cref{afig:clause-gadget-og}. In these figures, for sake of clarity, we use an X-shape to schematically represent that a pivot gadget is placed at that position. The regions $x_1$, $x_2$ and $x_3$ correspond to the literals of the clause and the regions 1, 2, 3 and 4 are used to put restrictions on these literals. The exact positions of the points of this gadget are as follows, taking the left top point as $(0,0)$:
        \[1, 2, 3, 4 \mapsto (0,0) , (0, -5) , (6, -5) , (6, 0)\]
        We place the point for $x_1$ at $(1,-3)$, $x_2$ at $(2, -4.6)$ and $x_3$ at $(5, -3)$.
        The construction used to connect the regions of these literals to variable gadgets and the exact placement of the pivot gadgets is explained in a later section. For now, it is only necessary to know that when a variable is false, then the part of the polyline going to the literal point can only go completely vertical through the pivot gadget at the bottom, or completely horizontal through the pivot gadgets at the sides. The regions of $x_1$, $x_2$ and $x_3$ are placed precisely so that in this situation there is exactly one realisation that does not intersect the pivot gadget. We shall refer to the position of a variable in this realisation as its \emph{false position} and say that it is \emph{covered} in a solution if it lies in the interior of the polyline defined by the corner points of the clause gadget. If the variable is true then the line can be placed at an angle through the pivot gadget, meaning that the literal point can be be realised in multiple different ways. Specifically we shall call the leftmost, bottommost and rightmost positions of $x_1$, $x_2$ and $x_3$ their respective \emph{true position}. The clause gadget simulates clauses of the Boolean formula according to the following lemma:
        \begin{lemma}
        For every choice of two of the three literals there exists a realisation that uncovers their false positions, and any realisation uncovers at most two false positions.
        \end{lemma}
        \begin{proof}
        The false positions of the literals $x_1$, $x_2$ and $x_3$ are $(1,-2)$, $(3, -4.6)$ and $(5, -2)$ respectively. Due to the placements of region 1 and 2, the only way to uncover position $(1, -2)$ is by placing the points of both of these regions at the rightmost positions, so at $(1,0)$ and $(1,-5)$. Similarly, the only way to uncover the false position of $x_3$ is by placing the points of region 3 and 4 at the leftmost positions $(5, -5)$ and $(5, 0)$. However, this means the line from region 2 to 3 is a horizontal line from $(1, -5)$ to $(5, -5)$. The false position of $x_2$ lies above this line which means it is not uncovered. Since this choice of points of regions 1 to 4 is the only one that uncovers the false positions of both $x_1$ and $x_3$, it is impossible to uncover all three false positions simultaneously.

        We can easily see that for any choice of two literals, we can uncover both false positions simultaneously by considering the separate cases. To uncover the false positions of both $x_1$ and $x_2$ we can choose for regions 1 and 2 the rightmost position, for region 3 the topmost position and for region 4 any position inside the region. A corresponding (part of a) weakly simple realisation is illustrated in \Cref{afig:clause-gadget-left}, where an orange line indicates that a literal is false and a blue line that it is true.  
        We have already described the case for $x_1$ and $x_3$ (see \Cref{afig:clause-gadget-down}) and the case for $x_2$ and $x_3$ (\Cref{afig:clause-gadget-right}) is symmetric to that of $x_1$ and $x_2$. The important observation is that for any weakly simple realisation, at least one of the variable points cannot be in its false position and must thus be set to true, just like in a disjunction.
        \end{proof}

        \begin{figure}[ht]
            \centering
            \subcaptionbox{\label{afig:clause-gadget-og}}[.5\textwidth]%
                {\includegraphics{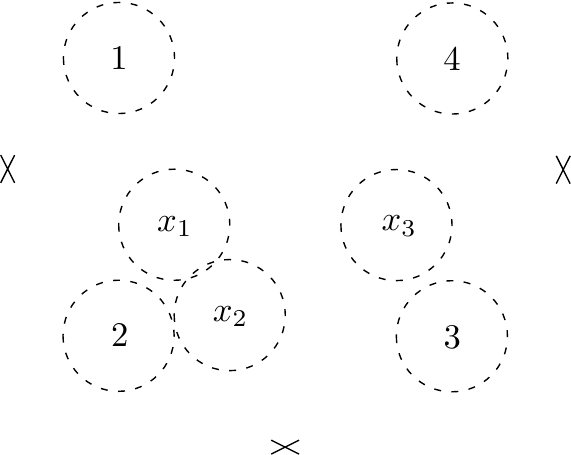}}%
            \subcaptionbox{\label{afig:clause-gadget-left}}[.5\textwidth]%
                {\includegraphics{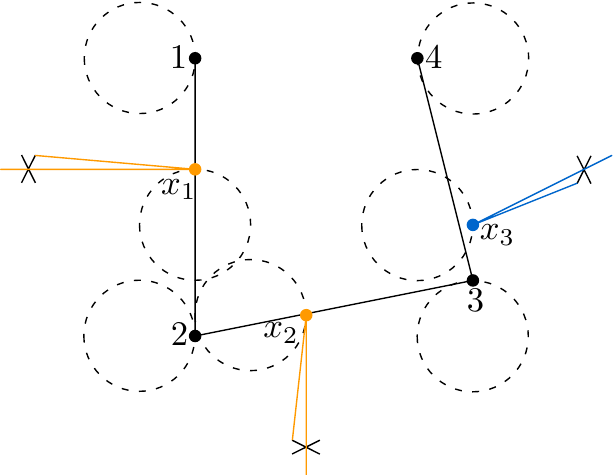}}\\%
            \subcaptionbox{\label{afig:clause-gadget-down}}[.5\textwidth]%
                {\includegraphics{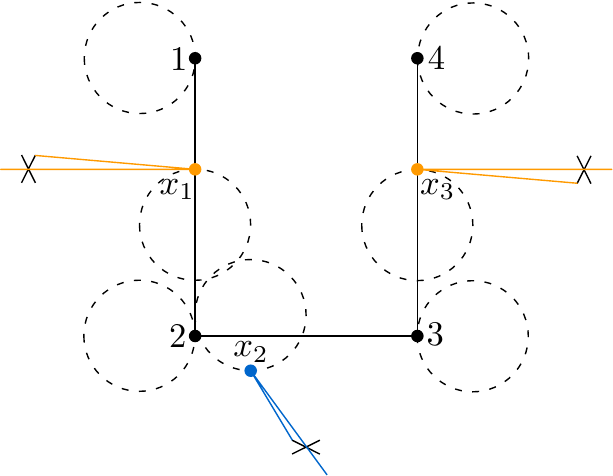}}%
            \subcaptionbox{\label{afig:clause-gadget-right}}[.5\textwidth]%
                {\includegraphics{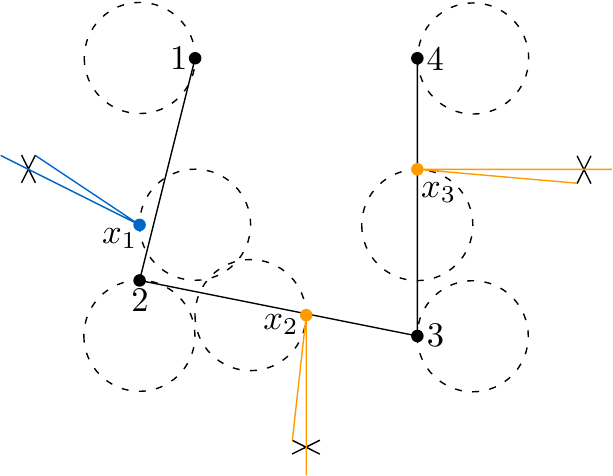}}%
            \caption{\subref{afig:clause-gadget-og} The clause gadget, with realisations where \subref{afig:clause-gadget-left} $x_3$, \subref{afig:clause-gadget-down} $x_2$, or \subref{afig:clause-gadget-right} $x_1$ is true.}
            \label{afig:clauses}
        \end{figure}
        
        \subparagraph*{Connecting variables and clauses.}
        For the reduction construction we still need a way to connect the variable gadgets to the clause gadgets correctly. For this we look at the shape of the induced variable-clause graphs we get for planar monotone 3SAT. We have seen an example before in \Cref{afig:plan3sat-vb}.
        Clauses are always connected to variables in an either downward or upward facing E-shape. To preserve the planarity of this graph, our reduction should follow this shape to a close enough degree. So the clause gadget is to be placed directly above one of the variable gadgets and the other two variable gadgets connect to it with a bend from the side.  
        
        We first look at the connection from a right variable to a clause. The wire gadget in this case involves three regions and two pivot gadgets as seen in \Cref{afig:wire-og}. Here region 1 is horizontally aligned with the variable gadget of the variable corresponding to the right literal of the clause. Based on \Cref{afig:variables}, we place these variable regions between region 2 and 6, which means depending on the state of the variable gadget our realisation has a horizontal line at the top (orange in \Cref{afig:wire-og}) if the state is false and at the bottom (blue) if the state is true. We denote the position of region 1 as $(0,0)$. Region three is the right literal region of the clause gadget. The corresponding two states of the clause gadget are drawn in blue and orange. The position of this region is determined by the clause gadget and it should not depend on the wire gadget.    
        Let region 3 to be positioned at $(-a, b)$ for some $a,b > 0$. For our wire gadget, we place region 2 at position $(1,b+1)$ and place two pivot gadgets at $(1-\frac{a}{2}, b+1)$ and $(0, 1 + \frac{b}{2})$, oriented appropriately. We stretch the variable-clause graph without changing its validity, so that $a$ and $b$ are large enough that the components of the wire gadget do not overlap.

        \begin{figure}[ht]
            \centering
            \subcaptionbox{\label{afig:wire-og}}[.25\textwidth]%
                {\includegraphics{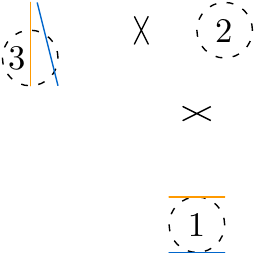}}%
            \subcaptionbox{\label{afig:wire-false}}[.25\textwidth]%
                {\includegraphics{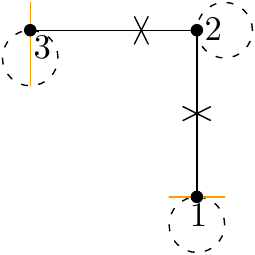}}%
            \subcaptionbox{\label{afig:wire-true-incorrect}}[.25\textwidth]%
                {\includegraphics{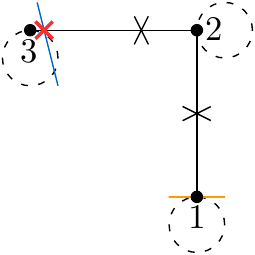}}%
            \subcaptionbox{\label{afig:wire-true}}[.25\textwidth]%
                {\includegraphics{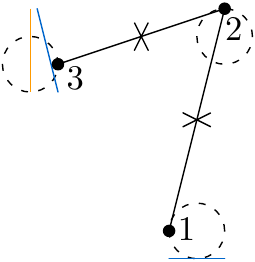}}%
            \caption{\subref{afig:wire-og} The right wire gadget to connect a variable to a clause. Realisations \subref{afig:wire-false} when both are in the false state, \subref{afig:wire-true-incorrect} where the variable is in its false state and the clause in its true state and \subref{afig:wire-true} with the variable gadget in its true state.}
            \label{afig:wires}
        \end{figure}
        
        \begin{lemma}
        If a variable gadget for $x_i$ is in its false state, any weakly simple realisation places $x_i$ in its false position.
        Otherwise $x_i$ can be placed in either its true or false position.
        \end{lemma} 
        \begin{proof}
        If the variable gadget is in its false state, its horizontal line will be at the top of region 1. The only choice of position for point in region 1 that prevents a crossing with this line is the very top of the region $(0, 1)$. Due to the pivot gadget, the line to region 2 must then be perfectly vertical in which case it only touches the region at position $(0,b+1)$, so the point of that region must be placed there. Similarly due to the second pivot gadget the point of region 3 must be placed at position $(a, b+1)$, which is its false position. This realisation can be seen in \Cref{afig:wire-false} and is the only one that does not directly cause intersections. If the variable gadget is in its false state but the clause gadget in its true state then this will always lead to a crossing. An example of this can be seen in \Cref{afig:wire-true-incorrect}. 
        If the variable gadget is in the true state the horizontal line is at the bottom of region 1 which means it does not constrain our choice of point in this region. Most notably, we can place a point at position $(-1, 0)$ which then propagates to the possibility of choosing point $(1, b+2)$ in region 2 and true position $(a+1, b)$ for region 3. This means there will be no crossings on this part of the polyline regardless of the state that the clause gadget is in (see \Cref{afig:wire-true}). Note that in contrast to the false realisation, the described true realisation is not the only one that is possible. Any choice of coordinates that leads to the point of the third region being to the right of the blue line and not causing any crossings is valid. 
        \end{proof}
        
        The wire gadget for connecting the left literal of a clause to its variable gadget is a mirrored version of the right wire gadget. The middle wire gadget is a simplified version needing only one pivot gadget, as this connection does not need to make a bend. All three versions of the wire gadget can be seen in effect in \Cref{afig:simple-example}.

        \subparagraph*{The reduction.}
        A general reduction method should be able to convert any monotone rectilinear layout of a planar monotone 3SAT formula into a sequence of unit disks for which there exists a weakly simple realisation if and only if the formula is satisfiable. 
        
        The general construction works as follows.
        First we place all gadgets as prescribed by the layout obtained from the planar monotone 3SAT instance.
        Next, we connect the gadgets in a planar manner.
        For this, we connect the variables from left to right as in \Cref{afig:simple-example}.
        We connect the gadgets in the top half of the construction by processing the clauses on the outer face (that lie above the $x$-axis) as follows, and connecting the results from left to right. 
        Connect the clause to its right wire (and its pivot gadgets), then recursively process and connect the clauses in the region enclosed by its right and middle wire, then connect the middle wire, then recursively process and connect the clauses in the region enclosed by the middle and left wire, and finally connect the left wire.
        The gadgets in the bottom half are connected in a symmetric manner.
        Due to the nature of our variable gadgets, making wire gadget connections from below automatically results in getting the negated value of the variable, which is in line with the definition of planar monotone 3SAT. Furthermore, the regions of the wire gadgets in the bottom half do not interfere (cause crossings) with those in the top half, since the horizontal line of the variable gadget ensures they remain separated even if their regions overlap.
        
        We give an example for the very simple case where we have only three variables and a single clause as seen in \Cref{afig:planarvcgraph-simple}. After construction we end up with a sequence of regions. In \Cref{afig:simple-example-down} we see a potential weakly simple realisation where the variable $x_2$ is set to true. Another example is shown in \Cref{afig:simple-example-left} where instead the $x_1$ variable gadget is true. It is for this example of course also possible for multiple variables gadgets to be in the true state simultaneously. 
        
        \begin{figure}[ht]
            \centering
            \subcaptionbox{\label{afig:planarvcgraph-simple}}[\textwidth]%
                {\includegraphics{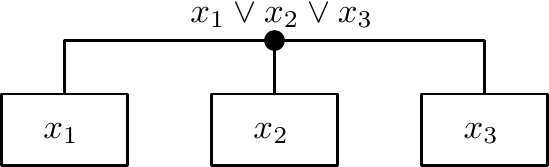}}\\%
            \subcaptionbox{\label{afig:simple-example-down}}[.5\textwidth]%
                {\includegraphics{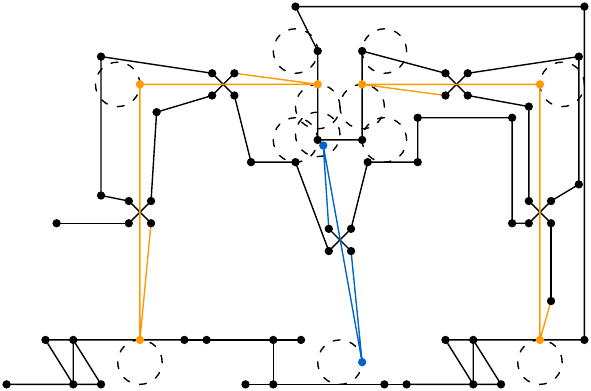}}%
            \subcaptionbox{\label{afig:simple-example-left}}[.5\textwidth]%
                {\includegraphics{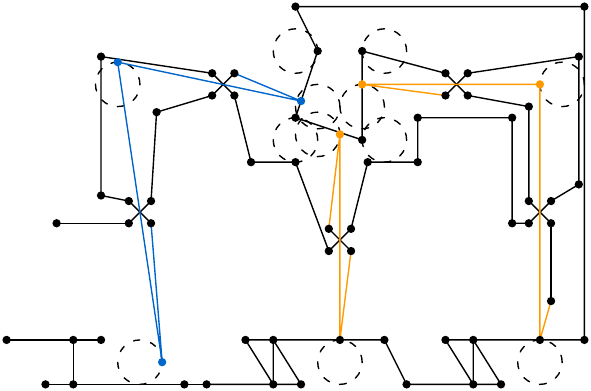}}%
            \caption{Example of the construction for variable-clause graph \subref{afig:planarvcgraph-simple}.
            The weakly simple realisation in \subref{afig:simple-example-down} sets $x_2$ to true and the realisation in \subref{afig:simple-example-left} sets $x_1$ to true.}
            \label{afig:simple-example}
        \end{figure}
        
        A potential issue could be that when we have nested clauses in the variable-clause graph, we have insufficient space to fully draw the clause gadget and its wires in the enclosure defined by the outer clauses. This is remedied by scaling the variable-clause graph, preserving planarity.
        
        Correctness of the reduction follows from correctness of the individual gadgets, for which we already argued when describing them. Furthermore, each gadget and connection uses a constant number of regions and their positions can be determined based on the variable-clause graph, which leads to the reduction being polynomial in size and time. This allows us to conclude the following:
        \begin{theorem}
        Given a sequence of unit disks, the problem of deciding if there exists a weakly simple realisation for this sequence is NP-hard. 
        \end{theorem}

        \begin{observation}
        In our construction, if there is a weakly simple realisation, then there is one that uses only the top, right, bottom or leftmost point of each imprecision region. Hence, if we use a different imprecision region that contains these four points, there will still be a weakly simple realisation if there was one initially. And if each region is a subset of the original, then if there is a weakly simple realisation in the second, then there is also one in the original. 
        \end{observation}
        
        The problem we considered can also be looked at from a different point of view. One could consider the center of each region in the sequence to be a point of an input polyline and the regions themselves denoting how we are allowed to perturb these points. In that case the setting of unit disks that we discussed can be seen as being allowed to move the point unit distance under the $L_2$ norm.  Therefore, it might be interesting to consider the cases for the $L_1$ and $L_\infty$ norms. Under the $L_1$ norm the regions are diamond-shaped. By the above observation the problem remains NP-hard. Furthermore, we can consider the square-shaped regions of the $L_\infty$ norm as the $L_1$ norm if we rotate everything by 45 degrees and scale the gadgets appropriately. This means that the problem stays NP-hard in this case.
    
    \section{Vertical segments as regions} \label{app:vertsegs}
        Next to the previously discussed cases we consider one additional scenario: what if the regions are all vertical segments of the same size? We show that even here, the problem remains NP-hard. We can prove this again using a reduction from planar monotone 3SAT. Since the reduction is very similar to the one from \Cref{app:circles} we shall keep the description concise. In our reduction we will work with vertical segments of length 2 as it makes the coordinates clearer, but this can easily be adapted to unit-length vertical segments.
        
        \subparagraph*{Pivot gadget.}
        Like before, in our reduction we are going to want to force lines to go through certain points by using a pivot gadget. However, the placement of regions used for disks will not work for vertical segments, since in this case it is not possible to place the points for regions 2 to 6 (and $2'$ to $6'$) all on a vertical line. However, due to the vertical segments as regions only allowing a single choice for the x-coordinate of their point, a simpler version of the pivot gadget as seen in \Cref{afig:vert-x-gadget-og} can be constructed. It is the same as the pivot gadget for unit disks, except we replace regions 2, 3, 5 and 6 by two regions at the exact same position (2 and 4 in the figure) and do the same for the bottom component. Similar to before, the pivot gadget requires any line of a weakly simple realisation to pass through the exact center as seen in \Cref{afig:vert-x-gadget-straight}. There is however the difference that now it is not possible to rotate the gadget due to regions being restricted to vertical segments. 
    
            \begin{figure}[ht]
            \centering%
            \subcaptionbox{\label{afig:vert-x-gadget-og}}[.5\textwidth]%
                {\includegraphics{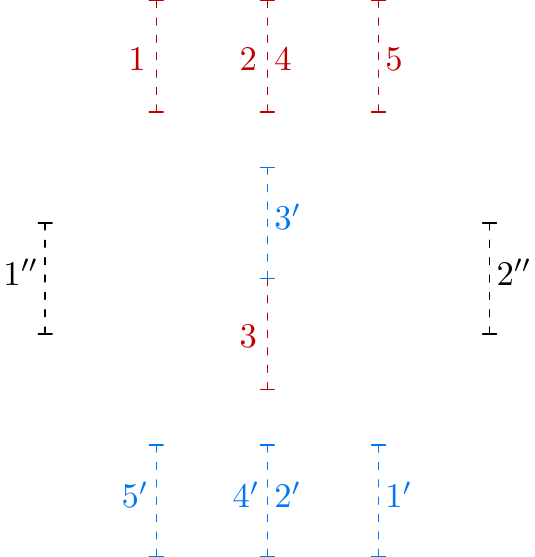}}%
            \subcaptionbox{\label{afig:vert-x-gadget-straight}}[.5\textwidth]%
                {\includegraphics{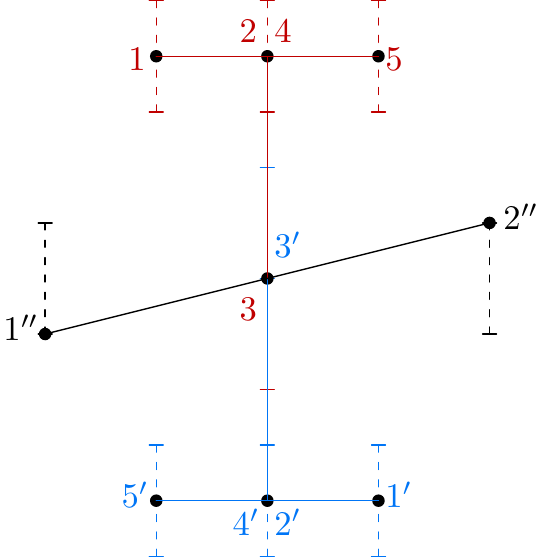}}%
            \caption{\subref{afig:x-gadget-og} The pivot gadget for vertical segments. \subref{afig:x-gadget-straight} A weakly simple realisation.}
            \label{afig:vert-X-gadget}
        \end{figure}
    
        \subparagraph*{Variable gadget.}
        The variable gadget from \Cref{app:circles} only makes use of the topmost and bottommost position of the regions. For this reason the variable gadget can stay exactly the same for the situation with vertical segments. 
        
        \subparagraph*{Clause gadget.}
        The clause gadget from \Cref{app:circles} fundamentally makes use of both the vertical and horizontal extremes of the disk regions. This means that we need to define a new clause gadget that does work for vertical segments. The new clause gadget places the regions at the following positions in order: 
        \[1,\ldots,6 \mapsto (0,0) , (0, -3) , (4, -5) , (6, -5) , (10, -3) , (10, 0)\]
        With two pivot gadgets placed with centers at $(2, -4)$ and $(8, -4)$. The regions for the literals should be placed at $(1, -4)$, $(5, -6)$ and $(9, -4)$ for $x_1$, $x_2$ and $x_3$ respectively. A simplified visualisation of construction can be seen in \Cref{afig:vert-clause-gadget-og}. For the literal regions we force the point to be placed at its topmost position when the literal it corresponds to is false (orange points in the figure). Just like before, we refer to this position as the false position of the literal. Only if the literal is true is the point allowed to be placed at its true position, the bottom of the region (blue points). The exact construction between variable and clause gadget to make this construction work is explained in a later part. In the same vein as before, the clause gadget can only uncover the false positions of at most two literal regions at the same time. In \Cref{afig:vert-clause-gadget-left} we see this happening for $x_2$ and $x_3$, meaning that $x_1$ needs to be true to not cause crossings. In \Cref{afig:vert-clause-gadget-down} this is the case for $x_2$ and a third realisation exists for $x_3$. 
        
        \begin{figure}[ht]
            \centering        
            \subcaptionbox{\label{afig:vert-clause-gadget-og}}[.33\textwidth]%
                {\includegraphics{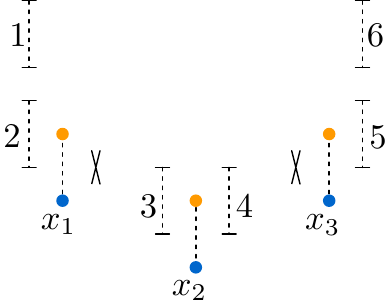}}%
            \subcaptionbox{\label{afig:vert-clause-gadget-left}}[.33\textwidth]%
                {\includegraphics{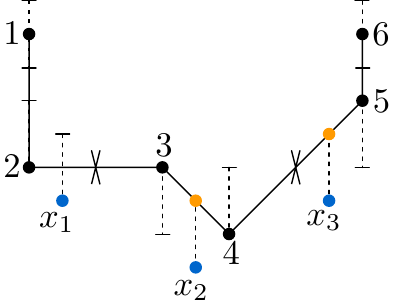}}%
            \subcaptionbox{\label{afig:vert-clause-gadget-down}}[.33\textwidth]%
                {\includegraphics{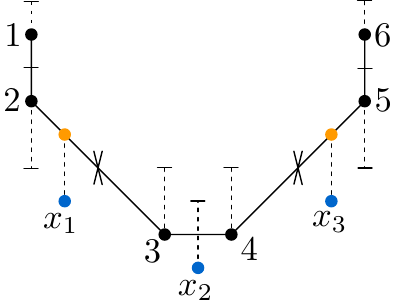}}%
            \caption{\subref{afig:vert-clause-gadget-og} The clause gadget for vertical segments, with realisations where \subref{afig:vert-clause-gadget-left} $x_1$ or \subref{afig:vert-clause-gadget-down} $x_2$ should be true.}
            \label{afig:vert-clause-gadget}
        \end{figure}
        
        \subparagraph*{Wire gadget.}
        Just like before, we need to define a wire gadget to connect literal regions at the clause gadget to the corresponding variable gadgets. The gadget consists of three regions (one of which the literal region) and two pivot gadgets. The wire is visualised in \Cref{afig:vert-wire-gadget} together with the two important weakly simple realisations. Region 1 is again placed at $(0,0)$ horizontally aligned with a variable gadget. Then assuming literal region 3 as positioned at $(-a, b)$ for $a,b > 0$, we get that region 2 must be placed at $(2, b)$ and the two pivot gadgets at $(1, \frac{b}{2})$ and $(\frac{-a}{2}+1, b)$. By placing these pivot gadgets exactly in the middle between two consecutive regions we get an alternating pattern in the solution polyline of the points lying at the very top or at the very bottom of their region, as can be seen in \Cref{afig:vert-wire-false} and \Cref{afig:vert-wire-true}. This idea can also be used to make a connection for the middle literal of the clause gadget where, in contrast to the case with unit disks, we do need an intermediate region and two pivot gadgets this time. 
        
        \begin{figure}[ht]
            \centering
            \subcaptionbox{\label{afig:vert-wire-og}}[.33\textwidth]%
                {\includegraphics{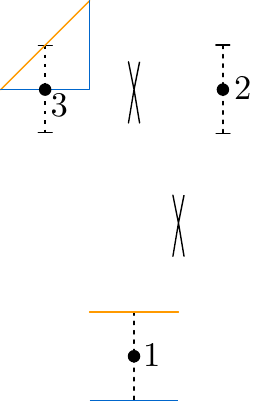}}%
            \subcaptionbox{\label{afig:vert-wire-false}}[.33\textwidth]%
                {\includegraphics{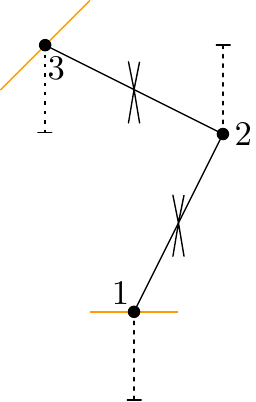}}%
            \subcaptionbox{\label{afig:vert-wire-true}}[.33\textwidth]%
                {\includegraphics{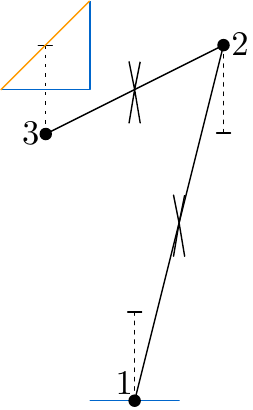}}%
            \caption{\subref{afig:vert-wire-og} The right wire gadget for vertical segments. Weakly simple realisations \subref{afig:vert-wire-false} where both the clause and variable gadget are in its false state and \subref{afig:vert-wire-true} where the variable gadget is in its true state.}
            \label{afig:vert-wire-gadget}
        \end{figure}
        
        \subparagraph*{The reduction.}
        Before looking at the full reduction, we can argue that the setting of vertical segments is in fact in NP.
        Our argument is similar to that used by Silveira et al.~\cite{noncrossingpaths}, where they argue for that a similar problem lies in NP using the fact that any solution can be represented as a solution to a system of linear constraints that can be guessed in nondeterministic polynomial time. 
        Consider a weakly simple realisation, and a corresponding simple perturbation.
        This simple perturbation induces an order type: for any vertex and any edge, record whether that vertex lies on, above, or below the supporting line of the edge.
        This order type induces a polynomial size system of linear constraints, and by picking a sufficiently small perturbation, we can ensure that the weakly-simple realisation lies in the closure of the solution space of this system of constraints.
        We can in nondeterministic polynomial time guess this system of linear constraints, and find a realisation in its closure.
        This realisation is then used as input to one of the algorithms~\cite{recogwspol, detectwspol} to verify that the realisation is weakly simple.
        Since the output of these algorithms depend only on the order type, they will report that the realisation is indeed weakly simple.
        \Cref{alem:segmentNP} follows.
    
        \begin{theorem}\label{alem:segmentNP}
        Given a sequence of unit-length vertical segments, deciding if there exists a weakly simple realisation for these regions is in NP. 
        \end{theorem}
        
        With the gadgets adapted to work for when the regions are vertical segments, the remainder of the reduction stays the same. We convert a monotone rectilinear layout of a formula to an instance of our problem by first constructing all variable gadgets, then all clauses and wires of the top half in order from right to left and finally all clauses and wires of the bottom half from left to right. 
        \begin{theorem}
        Given a sequence of unit-length vertical segments, the problem of deciding if there exists a weakly simple realisation for these regions is NP-complete. 
        \end{theorem}
        
    \section{Discussion}
        We showed that deciding if an imprecise polygon has a weakly simple realisation is NP-hard when each region is a unit disk, unit square, or vertical unit segment.
        By growing each region slightly, the same follows for deciding whether a simple realisation exists.
        Whereas the case of vertical segments is NP-complete, it is unclear whether the problems of \Cref{app:circles} lie in NP, as it is unclear whether polynomially many bits suffice to encode a solution.
        The settings of \Cref{app:circles} lie in the complexity class $\exists\mathbb{R}$, and we wonder if these settings are $\exists \mathbb{R}$-hard.

\bibliography{references}

\end{document}